\documentclass[11 pt]{article}
\renewcommand{\baselinestretch}{1.47}
\date{}

\author{Javad Heydari \and Ali Tajer\thanks{Authors are with the Electrical, Computer, and System Engineering Department, Rensselaer Polytechnic Institute, Troy, NY 12180.}}
\usepackage{graphicx,subfigure}
\usepackage{amsmath,amssymb}
\usepackage[export]{adjustbox}
\usepackage{cite}
\usepackage{url}
\usepackage{booktabs}
\usepackage{hhline}
\usepackage{dsfont}
\usepackage{bm}
\usepackage {tikz}
\usetikzlibrary {positioning}
\usepackage{epstopdf}
\usepackage{ifpdf}
\usepackage{array}
\usepackage{amsthm}
\usepackage[font=footnotesize,singlelinecheck=false]{caption}
\usepackage{hyperref}

\newtheorem{theorem}{Theorem}

\newtheorem{definition}{Definition}

\def \sH{{\sf H}}
\def \T{{\sf T}}

\def \H{{\sf H}}

\newcommand{\dff}{\stackrel{\scriptscriptstyle\triangle}{=}}
\newcommand{\bbe}{\mathbb{E}}

\def \bX{\boldsymbol{{ X}}}
\def \bM{\boldsymbol{{ M}}}
\def \blm{\boldsymbol{{ m}}}
\def \bH{\boldsymbol{{ H}}}

\def \bU{\boldsymbol{{ U}}}
\def \bV{\boldsymbol{{ V}}}
\def \bA{\boldsymbol{{ A}}}
\def \bI{\boldsymbol{{ I}}}

\def \by{\boldsymbol{{ y}}}

\def \bp{\boldsymbol{{ p}}}
\def \br{\boldsymbol{{ r}}}
\def \ba{\boldsymbol{{ a}}}
\def \bB{\boldsymbol{{ B}}}

\def \bn{\boldsymbol{{ n}}}
\def \bs{\boldsymbol{{ s}}}
\def \bt{\boldsymbol{{\theta}}}

\def \A{{\cal A}}

\def \G{{\cal G}}
\def \L{{\cal L}}
\def \B{{\cal B}}

\def \N{{\cal N}}

\def \E {{\cal E}}

\def \cT {{\cal T}}

\def \R {{\cal R}}

\DeclareMathOperator*{\argmax}{arg\,max}
\DeclareMathOperator*{\argmin}{arg\,min}

\begin{document}

\title{Quickest Localization of Anomalies in Power Grids: \\ A Stochastic Graphical Framework}

\maketitle

\begin{abstract}

Agile localization of anomalous events plays a pivotal role in enhancing the overall reliability of the grid and avoiding cascading failures. This is especially of paramount significance in the large-scale grids due to their geographical expansions and the large volume of data generated. This paper proposes a stochastic graphical framework, by leveraging which it aims to localize the anomalies with the minimum amount of data. This framework capitalizes on the strong correlation structures observed among the measurements collected from different buses. The proposed approach, at its core, collects the measurements sequentially and progressively updates its decision about the location of the anomaly. The process resumes until the location of the anomaly can be identified with desired reliability. We provide a general theory for the quickest anomaly localization and also investigate its application for quickest line outage localization. Simulations in the IEEE 118-bus model are provided to establish the gains of the proposed approach.

\end{abstract}


\vspace{-0.1 in}
\section{Introduction}
\label{sec:intro}


Due to the large-scale and strong inter-connectivities in the power grid, any fault or failure can transcend its  realm and disrupt operations in other parts of the grid as well. This can potentially cause disruption of service and destabilize grid operations. Therefore, real-time monitoring of a grid, consisting of generators, transmission lines, and transformers, is of paramount importance in securing reliable power delivery. Specifically, agile detection and localization of system failures facilitate mitigating the disruptive impacts the failure can cause to the network, and prevent anomalous events that can lead to failures in larger scales. The introduction of advanced measurement devices such as phasor measurement units (PMUs) has enabled collecting real-time synchronized data from the entire network, which allows the operators to dynamically observe the status of the system and detect and even localize potential failures.

Transmission lines are constantly exposed to various kinds of disturbances such as equipment malfunctioning and natural disasters. While the power system is designed to operate under single or multiple contingencies, the monitoring task should identify those contingencies quickly to prevent overload in one section of the grid which may lead to the cascade of events and a major blackout. Detecting such contingencies and anomalies when they occur, and localizing them accurately can expedite the repair of the faulty components, speed up restoration of the grid, reduce outage time, and improve power system reliability~\cite{Liao}. Hence, anomaly detection and localization have been investigated extensively in the existing literature under different settings and objectives. Detection of anomalies, identifying their location, and specifying the type of the anomalous events are the main objectives of fault analysis in power grids. In this paper we develop a stochastic graphical framework for modeling the bus measurements, and devise data-adaptive data-acquisition and decision-making processes for reliably detecting and localizing the anomalous events with the fewest number of measurements.  This is motivated by lowering the required communication and reducing the computational complexity  and delay of decision-making.

Analyzing anomalous events can be categorized into two broad classes according to the type of information used. In one direction, detecting and localizing events are based on the changes in the impedance of the corresponding transmission lines which are leveraged to detect and localize the event by evaluating voltage and current measurements. The available data in this method plays a critical role in the complexity and detection accuracy. Local approaches, according to the number of terminals from which measurements can be taken, are categorized into single-end~\cite{Takagi,Eriksson,Kawady,Pereira}; double-end~\cite{Kezunovic,Dalcastagne,Liao2009,Izykowski,Apostolopoulos}; and multi-end algorithms~\cite{Nagasawa,Tziouvaras,Manassero,Funabashi,Brahma}. The systems that use values measured in both line terminals give more exact results than those that only use values measured in one terminal. Nevertheless, in double-end approaches, measurements require synchronization which makes data acquisition more complex.

In another direction anomalous events are studied based on the high frequency contents of the signal propagated in the network under an event~\cite{McLaren,Ibe,Magnago,Evrenosoglu,Spoor,Jafarian,Korkali}. In these approaches, signature waves are sent along the transmission lines, and the traveling durations are determined by leveraging the correlation between forward and backward waves. Such time durations explicitly determine the distance from that terminal to the anomalous point. This class of localization techniques are insensitive to fault type, fault resistance, and source parameters of the system, and are independent of the equipments installed in the network. For the arrival times, feature extraction techniques such as wavelet transform are leveraged to distinguish between the normal signal and the one containing high frequency components. Feature extraction techniques combined with classification methods such as neural networks can also be used for anomaly detection~\cite{vasilic2001new}, their classification~\cite{Kashyap},~\cite{Lin2001}, and localization~\cite{Jiang2012,Silva,Gracia}.

All the aforementioned studies utilize a static monitoring mechanism for analyzing the events, i.e., pre-specified locations of the grid are monitored continuously. It implies that a sufficiently large number of measurements is required to ensure reliable detection and localization of anomalies. Despite the effectiveness, such approaches can become inefficient in large-scale networks that are expanded over a large geographic area, due to the costs associated with collecting and processing large volumes of data. To circumvent this issue, a stochastic graphical framework is developed in this paper to model the measurements collected from the grid. Generated measurements at different buses around the network follow a certain correlation structure which depends on the topology of the network and the status of different transmission lines. Under an anomalous event, this correlation model changes to the one that reflects the location of the event. This framework is leveraged to minimize the number of measurements required to ensure that all the events can be localized with a target reliability through designing a coupled data-acquisition and decision-making process. This leads to minimizing the amount of data required for localizing the fault. Specifically, it develops a stochastic graphical model in which the connectivities in the graph are modeled based on the grid parameters and capture the correlation among the measurements reported by neighboring buses. By properly leveraging such correlation, the quality of the information provided by different measurement units are quantified. This enables devising a data-adaptive information-gathering process, which can dynamically form an estimate about the location of the potential event and, accordingly, measure the buses that are most informative about the anomaly.

This paper designs a {\em quickest} coupled data-acquisition and decision-making strategy for detecting and localizing anomalous events in transmission lines. The purpose of such a strategy is to detect and localize the anomalies with the minimum number of measurements, while satisfying a target reliability for the decisions. In contrast to the non-adaptive strategies, which collect the data according to a pre-specified rule and in one shot, the proposed adaptive approach gradually and progressively focuses its sampling resources on the areas in the network which are most likely to contain the anomalous line(s). Specifically, this approach starts by taking rough measurements from potentially anywhere in the network, and based on the collected measurements, dynamically and over time it eliminates the regions considered to behave normally, and further scrutinizes those that are stronger candidates for behaving anomalously. Designing such strategies involves balancing a tension between the accuracy and agility of the decision, as two opposing performance measures. Specifically, achieving a higher quality in decision necessitates collecting more data, which in turn penalizes the delay of the process. This data-acquisition and decision-making strategy involves making dynamic decisions at each time about 1) what set of measurement units to be measures, and 2) whether a reliable decision can be formed based on the collected data, or more measurements are still needed. Under each anomaly, it is assumed that the network remains connected and it settles down to a steady-state quickly. Furthermore, the collected data are assumed to bear no measurement noise or data injection attacks. We first review the preliminaries on the stochastic graphical model and anomaly detection in Section~\ref{sec:pre}. Quickest anomaly detection and localization is formalized in Section~\ref{sec:problem}. In Section~\ref{sec:sol} we present the theory for characterizing the data acquisition and decision-making processes as well as the general framework for anomaly detection, and the associated algorithms for implementing the optimal decision rules. Finally, in Section~\ref{sec:sim} we apply the designed algorithm to the problem of line outage detection and localization as a special anomaly detection problem. It is noteworthy that the line outage detection in power grids is investigated extensively in the existing literature. When a transmission line is in outage, it is assumed that the tripping log of its associated relay is not available or accessible. Hence, the localization of the outage should be performed based on the phasor measurements from different buses. When all the measurements are available, exhaustive search for detecting single line outage events is studied in~\cite{Abur} and~\cite{Tate1}, and computing line outage distribution factors for detecting multiple line outages is studied in~\cite{LODF1} and~\cite{LODF2}. In~\cite{Veeravalli}, a quickest change point detection approach is deployed that monitors the network sequentially in order to detect a persistent outage and identify its location. Joint outage detection and state estimation are considered under the Bayesian setting in~\cite{Zhao14}. In~\cite{He:TSG11}, measurements are modeled as a Gauss-Markov random field (GMRF) and outage detection is performed by approximating the covariance matrix of the measurements. The study in~\cite{Zhu:PS12} formulates outage detection as a sparse signal recovery problem and applies compressive sensing tools for outage detection. All these studies utilize the measurements from all the buses and their performance degrades significantly when a subset of measurements are available. To address this issue, the study in~\cite{Wu} develops an algorithm based on the ambiguity group theory for localizing the outage event. Another approach is to estimate the unobserved PMU data prior to performing detection~\cite{Maymon}. The optimal static PMU selection for minimizing the error probability in outage detection over all possible outage events is studied in~\cite{Zhao:PES12,Zhao:PES13,Wu15,Kim}.

\vspace{-0.05 in}
\section{Preliminaries}
\label{sec:pre}

\subsection{Background on Markov Random Fields}

A Markov random field (MRF) is a graphical model that encodes certain dependency structures among a collection of random variables. Given an undirected graph $\mathcal{G}=(\mathcal{B},\mathcal{E})$ with $N$ nodes $\mathcal{B}\dff\{1,2,\dots,N\}$, the set of random variables $\boldsymbol{\theta}\dff\{\theta_1,\dots,\theta_N\}$ form a Markov random field with respect to $\G$ if they satisfy the global Markov property. To formalize this property, for any given set $A\subseteq \{1,\dots,N\}$ we define $\theta_A\dff \{\theta_i\;:\; i\in A\}$. We also say that set $C$ separates disjoint sets $A$ and $B$ if any path starting in $A$ and terminating in $B$ has at least one node in $C$.
\begin{definition}[Global Markov property]
The set of random variables $\boldsymbol{\theta}\dff\{\theta_1,\dots,\theta_N\}$ satisfies the global Markov property associated with graph $\mathcal{G}=(\mathcal{B},\mathcal{E})$ if and only if for any two disjoint subsets $A,B\subseteq\mathcal{B}$ and a separating subset $C\subseteq\B$, random  variables $\theta_A$ and $\theta_B$ are conditionally independent given $\theta_C$, i.e.,
\begin{align}\label{eq:global}
\mathbb{P}(\theta_A \;|\; \theta_B , \theta_C)=\mathbb{P}(\theta_A \;|\; \theta_C) \ .
\end{align}
\end{definition}
\noindent Random variables satisfying the global Markov property also satisfy the following weaker Markov property.
\begin{definition}[Local Markov property] A random variable is conditionally independent of all other random variables, given its neighbors, i.e., 
\begin{align}\label{eq:local}
\mathbb{P}(\theta_u \;|\; \theta_v , \theta_{\N_u})=\mathbb{P}(\theta_u \;|\; \theta_{N_u})  \quad \forall v\notin (\N_u\cup{u}) \ ,
\end{align}
where $\N_u$ denotes the set of neighbors of $u$, i.e., 
\begin{align*}
\N_u\dff \{v\in\B\;:\;(u,v)\in\E \}\ .
\end{align*}
\end{definition}

\subsection{Statistical Model of Bus Measurements}

Studies in \cite{He:TSG11}, \cite{Sedghi14}, and \cite{Sedghi15} show that the statistical relationship among the measurements collected from different buses across the grid can be modeled effectively by a GMRF. Based on this model, grid topology determines the graph underlying the GMRF, such that the buses correspond to the vertices of the graph and the lines constitute the edges. This model relies on the observation that the second-neighbor correlations are dominated by those of the immediate neighbors \cite{Sedghi14}.

To formalize this connection, consider a power grid consisting of $N$ buses, abstracted by graph $\G=(\B,\E)$, where $\B\dff\{1,\dots,N\}$ denotes the set of buses and $\E\subseteq\B\times\B$ represents their connectivities such that $(i,j)\in\E$ if buses $i,j\in\B$ are directly connected by a line. We define $\theta_i$ and $p_i$ as the voltage phasor angle and the injected active power at bus $i\in\B$. By defining $x_{ij}$ as the reactance of the line connecting buses $i$ and $j$, from the DC power flow model we have \cite{Abur:book}:
\begin{align}\label{eq:linear}
p_i=\sum_{j\in\N_i}\Big(\frac{\theta_i-\theta_j}{x_{ij}}\Big)\ .
\end{align}
Hence, defining $\bp\dff[p_1,\dots,p_N]^T$ and $\bt\dff[\theta_1,\dots,\theta_N]^T$ provides
\begin{align}\label{eq:matrix}
\bp=\bH\cdot\bt \ ,
\end{align}
where $\bH\in\mathbb{R}^{N\times N}$ is the weighted Laplacian matrix of the connectivity graph defined as
\begin{equation}\label{eq:H}
\bH[ij]=\left\{\begin{array}{ll}
\vspace{1mm}
\sum_{(i,\ell)\in\E}\frac{1}{x_{i\ell}} & \text{if } i=j \\ 
\vspace{1mm}
-\frac{1}{x_{ij}} & \text{if } (i,j)\in\E \\ 
0 & \text{Otherwise}
\end{array}\right. .
\end{equation} 
Furthermore, from \eqref{eq:linear} it follows that $\theta_i$ can be represented as
\begin{align}\label{eq:phasor}
\theta_i=\sum_{j\in\N_i}r_{ij}\theta_j+\beta_i p_i \ ,
\end{align}
where we have defined
\begin{align}\label{eq:r}
\beta_i  \;\dff\; \bigg(\sum_{(i,j)\in\E} \frac{1}{x_{ij}}\bigg)^{-1} \ ,\quad\text{and }\quad r_{ij} & \;\dff\; \frac{\beta_i}{x_{ij}} \ .
\end{align}
Equation~\eqref{eq:phasor} indicates that $\theta_i$  depends only on the voltage phasor angles at its immediate neighbors. By accounting for the random disturbances in the system as well as the uncertainties associated with load profiles, the aggregate injected power at different buses can be modeled as independent random variables~\cite{Dopazo} and~\cite{Schellenberg}. This assumption in conjunction with~\eqref{eq:phasor} shows that the set of voltage phasor angles $\{\theta_i:\ i\in\B\}$ satisfy the local Markov property.

We note that, by construction, matrix $\bH$ is rank-deficient, which causes ambiguity for the solution of $\bt$ in \eqref{eq:matrix}. To fix this ambiguity, one bus is selected as reference with its phasor angle set to zero and the phasor angles of all other buses denote their differences relative to the reference bus. By removing the row and column corresponding to the reference bus, the remaining $(N-1)\times(N-1)$ matrix ￼$\bH$ has full rank. In the remainder of this paper, when referring to the Laplacian matrix of the network, we always mean the modified full-rank one.

\subsection{Anomalous Events}

\begin{figure}
\centering
\includegraphics[width=2.5in]{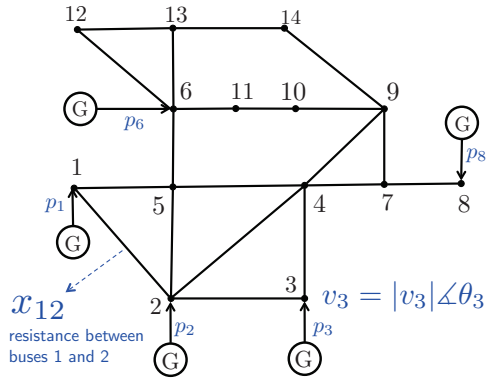}
\renewcommand{\figurename}{Fig.}
\caption{IEEE $14$-bus standard system with $20$ transmission lines and $5$ generators. The transmission line between bus $(i,j)$ has reactance $x_{ij}$.}
\label{fig:model}
\end{figure}

The grid consists of $L\dff|\E|$ transmission lines, where the set of lines is denoted by $\L\dff\{1,\dots,L\}$. We are interested in detecting and localizing anomalous events in the transmission lines. Any change in the reactance of transmission lines that does not conform with the expected patterns is considered to be an anomalous event and, there exist $(2^L-1)$ such possible events, each one corresponding to one combination of lines experiencing anomaly. These events, in practice, occur with different frequencies, and have different disruptive effects on grid operations. Furthermore, we only consider the events that keep the underlying post-event graph connected. This precludes considerable changes between pre-event and post-event bus power injections. We define $M$ as the number of events that represent the most critical ones, which should be localized in the quickest fashion. Accordingly, we define $\R=\{R_1,\dots,R_M\}$ as the set of such events, where $R_k\subseteq\L$ contains the indices of the lines experiencing anomaly under event $k\in\{1,\dots,M\}$. Additionally, event $R_0$ is reserved to signify the event under which all lines are normal. Dynamically determining the state of the grid, and localizing the anomaly, when the grid is deemed to be anomalous, can be abstracted as dynamically deciding which event $R_k\in\R$ represents the model of the grid. By denoting the true event by $\T\in\R$, detecting and localizing anomalies can be cast as the following multi-hypothesis testing problem:
\begin{align}\label{eq:problem1}
\sH_k:\quad \T=R_k\ , \quad \mbox{for}\ k\in\{0,\dots,M\}\ .
\end{align}
When an anomaly occurs, network connectivity profile changes. We denote the connectivity graph of the grid and the reactance of the line connecting buses $i$ and $j$ under event $R_k$ by $\G_k(\B,\E_k)$ and $x^k_{ij}$, respectively, corresponding to which for $k\in\{0,\dots,M\}$ we define matrix $\bH_k$ such that
\begin{equation}\label{eq:H_k}
\bH_k[ij]\dff\left\{\begin{array}{ll}
\vspace{1mm}
\sum_{(i,\ell)\in\E_k}\frac{1}{{x}^k_{i\ell}} & \text{if } i=j \\ 
\vspace{1mm}
-\frac{1}{x^k_{ij}} & \text{if } (i,j)\in\E_k \\ 
0 & \text{Otherwise}
\end{array}\right. .
\end{equation} 
Hence, the multi-hypothesis model in~\eqref{eq:problem1} can be expressed as
\begin{align}\label{eq:problem2}
\sH_k:\quad \bt=\bB_k\cdot\bp\ , \quad\mbox{for}\ k\in\{0,\dots,M\}\ ,
\end{align}
where we have defined $\bB_k\dff\bH_k^{-1}$. Under each anomalous event, $\bt$ follows a distinct correlation structure governed by the associated topology and line reactances of the network, which is imposed through matrix $\bB_k$. Due to the massive scale of power networks, collecting measurements from all the buses incurs prohibitive sensing and processing costs. Hence, we devise a data-adaptive decision-making framework that can form arbitrarily reliable decisions about the state of the grid with {\em minimal} number of measurements.

\section{Quickest Localization of Anomalies}
\label{sec:problem}

In this section, we formalize a sequential data-acquisition and decision-making process to collect measurements of voltage phasor angles and use these measurements to localize the anomalies, when one is deemed to exist, with the {\em fewest} number of measurements. This is motivated by reducing the costs associated with data-acquisition, communication, and processing, especially in large-scale grids.
This data collection and decision-making mechanism is constructed based on the premise that when a specific anomaly occurs, it affects the measurements from different buses with varying degrees. For instance, when the line connecting bus $i$ and $j$ is in outage, its effects on the measurements generated at buses $i$ and $j$ can be more than those of a remote bus. By capitalizing on such discrepancies among the level of information provided by different buses, the proposed sampling procedure progressively forms a decision about the likely events, and takes measurements from the buses that are expected to be more informative about these events.

The data-acquisition process sequentially collects $\ell$ measurements at-a-time from $\ell$ different buses. The process continues until time $\tau\in\mathbb{N}$, as the stopping time of the process, at which point it terminates and a decision about the underlying event is formed. For modeling the dynamic decisions about the buses to be observed at time $t$ we define the selection function $\psi(t)\dff[\psi(t,1),\dots,\psi(t,\ell)]$, which captures the indices of $\ell$ buses to be measured at time $t\in\{1,\dots,\tau\}$. We denote the vector of measurements collected at time $t$ by $\bt(t)\dff[\theta(t,1),\dots,\theta(t,\ell)]$ where $\theta(t,i)$ is the measurement collected from bus $\psi(t,i)$. Accordingly, we denote the vector of observed buses and their corresponding measurements up to time $t$ by $\psi_t$ and $\bt_t$, respectively, i.e.,
\begin{align}
\psi_t &\dff [\psi(1),\dots,\psi(t)]^T\ \;\;
\mbox{and} \;\; \bt_t &\dff [\bt(1),\dots,\bt(t)]^T\ .
\end{align}
Finally, we define $\delta\in\R$ as the decision rule at the stopping time. The quality of decision at the stopping time is captured by the decision error probability, i.e.,
\begin{align}
\mathsf{P_e}&=\mathbb{P}(\delta\neq \T)=\sum_{i=0}^M\mathbb{P}(\T=i)\sum_{j\neq i}\mathbb{P}(\delta=j \;|\; \T=i)\ .
\end{align} 
Hence, the optimal sampling strategy, which aims to form the quickest decision subject to maintaining a target decision quality is obtained as the solution to the following optimization problem.
\begin{align}\label{eq:Opt2}
\underset{\tau\, ,\, \delta \, ,\, \psi_\tau}{\text{minimize}}  \;\;\mathbb{E}\{\tau\} \; 
\;\;\;\text{subject to } \;\;\; \mathsf{P_e}\leq\beta
\end{align}
where $\beta\in(0,1)$ controls the probability of erroneous decisions.

\section{Optimal Decision Rules}
\label{sec:sol}

The optimal sampling strategy involves dynamically selecting the buses to be monitored and deciding the time to stop the process. In the next subsections we characterize the bus selection function $\psi(t)$, the stopping time $\tau$, and the final decision rule $\delta$. We remark that these rules, collectively, satisfy asymptotic optimality guarantees and solve the quickest detection problem of interest formalized in \eqref{eq:Opt2}. In the sequel, we assume that the total number of lines that can be concurrently anomalous is upper bounded by $\eta_{\rm max}$.

\vspace{-0.1 in}
\subsection{Bus Selection Rule}

\subsubsection{Analysis}

The measurements from different buses are not equally informative about different events. Hence, dynamically selecting buses based on real-time data for measuring their voltage phasor angles has a critical role in striking an optimal balance between the decision quality and the quickness of the process, as formalized in \eqref{eq:Opt2}. In order to characterize the optimal bus selection rule $\psi(t)$, we start by establishing the relevant theoretical foundations, and then we provide the specific designs for $\psi(t)$. Solving the problem in \eqref{eq:Opt2} can be facilitated  by using the techniques in controlled sensing, and specifically the Chernoff rule~\cite{Chernoff1959}. According to the Chernoff rule, at each time $t$ we first make a maximum likelihood decision about the true model $R_k$ based on which we select the bus that reinforces this decision to be measured at time $(t+1)$. The information of each observation is quantified in terms of the Kullback-Leibler (KL) divergence between the distributions under various hypotheses. The main advantage of the Chernoff rule is low computational complexity. The main weakness, on the other hand, is that it can be suboptimal as it decouples the impact of the decisions made at each time on the future decisions. To circumvent this deficiency, we propose a new selection rule to incorporate the effect of each action on the future ones. This new decision rule, in general, involves an {\em exhaustive} search over all buses and can have prohibitive complexity, especially as the grid size grows. Nevertheless, we show that by properly leveraging the Markov structure, the computational complexity can be reduced significantly and it becomes as simple as that of the Chernoff rule. In order to prove these properties, we first focus on a binary setting, i.e., $\R=\{R_0,R_1\}$ and consider taking one sample at-a-time, i.e., $\ell=1$. Under the normal event $R_0$, we assume that the measurements form a GMRF with mean $\bar\bt$, which represents the empirical average of $\bt$ based on the historical data, and covariance matrix $(\bI-\bm{Q}_0)$, where the elements of $\bm{Q}_0=[r_{ij}]$ are defined in \eqref{eq:r}. Under the anomalous event $R_1$, on the other hand, we assume that the measurements form an alternative GMRF with a different covariance matrix $\bm{Q}_1$. For the simplicity in notations we assume $\bm{Q}_1=\bI$.  We define set $\mathcal{S}_t^i$ as a subset of unobserved buses prior to time $t$ that contain bus $i$, i.e.,
\begin{align}
\mathcal{S}_t^i\subseteq\B\setminus\psi_{t-1} \quad\text{and}\quad i\in\mathcal{S}_t^i \ .
\end{align}
Furthermore, at time $t$, and corresponding to each valid set $\mathcal{S}_t^i$ we assign the following two metrics to each bus $i$:
\begin{align}\label{eq:metric:GMRF1}
M_i^0(t,\mathcal{S}_t^i) \;=\; & \frac{1}{2} \sum_{j\in\psi_{t-1}} \log\frac{1}{1-r^2_{ij}}+r^2_{ij}(\Delta\theta_j^2-1) \nonumber\\
+\; & \frac{1}{2|\mathcal{S}_t^i|}\sum_{j\in\mathcal{S}_t^i}\log\frac{1}{1-r_{ij}^2}\ , \\
\label{eq:metric:GMRF2}
\text{and}\ \ M_i^1(t,\mathcal{S}_t^i) \;=\; & \frac{1}{2} \sum_{j\in\psi_{t-1}} \log(1-r^2_{ij})+\frac{r^2_{ij}(\Delta\theta_j^2+1)}{1-r^2_{ij}} \nonumber\\ 
+\; & \frac{1}{2|\mathcal{S}_t^i|}\sum_{j\in\mathcal{S}_t^i}\log(1-r_{ij}^2)+\frac{2r_{ij}^2}{1-r_{ij}^2}\ ,
\end{align}
where we have defined $\Delta\theta_i\dff\theta_i-\bar\theta_i$ for $i\in\{1,\dots,N\}$.
Based on these definitions, when the maximum likelihood decision about the true model at time $(t-1)$ is $\H_0$, at time $t$ we select bus
\begin{align}\label{eq:op}
\psi(t)=\argmax_{i\notin\psi_{t-1}}\ \max_{\mathcal{S}_t^i}\ M_i^0(t,\mathcal{S}_t^i)\ .
\end{align}
Similarly, when the maximum likelihood decision about the true model at time $(t-1)$ is $\H_1$, we select
\begin{align}\label{eq:op2}
\psi(t)=\argmax_{i\notin\psi_{t-1}}\ \max_{\mathcal{S}_t^i}\ M_i^1(t,\mathcal{S}_t^i)\ .
\end{align}
Determining the selection function in \eqref{eq:op} and \eqref{eq:op2} is computationally prohibitive as it involves an exhaustive search over all the possible subsets of unobserved buses. However, our analyses demonstrate that the complexity of such an exhaustive search over GMRFs can be reduced substantially by analytically proving that the optimal group of the buses to be measured belong to a small subset of buses. Specifically for each node $i$, the choice of the set $\mathcal{S}_t^i\setminus\{i\}$ is limited to the subset of the unobserved neighbors of $i$, i.e.,
\begin{equation}
\mathcal{U}_t^i\dff\N_i\setminus \psi_{t-1}\ .
\end{equation}
This indicates that for determining which node to select at each time in a GMRF it is sufficient to consider a shorter future for each node, while in general, we have to decide based on all the remaining nodes. The cardinality of the set of subsets of $\mathcal{U}_t^i$ is significantly smaller than that of unobserved nodes, which translates into significant reduction in the complexity of characterizing the  optimal selection functions. This observation is formalized in the following theorem.
\begin{theorem}\label{thm1}
At each time $t$, for all valid sequences $\mathcal{S}_t^i$ and for $u\in\{0,1\}$ we have
\begin{align}
\argmax_{i\notin\psi_{t-1}}\ \max_{\mathcal{S}_t^i} {M_u^i(t,\mathcal{S}_t^i)}=\argmax_{i\notin\psi_{t-1}}\max_{\mathcal{S}_t^i\subseteq\mathcal{U}_t^i} {M_u^i(t,\mathcal{S}_t^i)}\ .
\end{align}
\end{theorem}
\begin{proof}
See Appendix~\ref{App:thm1}.
\end{proof}
\noindent This theorem states that it suffices to search over the neighbors of each bus to find the bus that provides the most relevant information about the underlying event. The structure of the metric for each bus depends on the joint distribution of voltage phasor angles. Next, we show that the selection rule that only searches over the neighbors of one node achieves asymptotic optimality as the size of the network grows and the frequency of erroneous decisions tends to zero. This statement is formalized in the following theorem. 
\begin{theorem}\label{thm2}
For the quickest anomaly detection and localization problem given in \eqref{eq:Opt2}, the selection functions in \eqref{eq:op} and \eqref{eq:op2} achieve asymptotic optimality as $\beta$ approaches zero, i.e., for $i\in\{0,1\}$
\begin{align}
\lim_{\beta\rightarrow 0}\frac{\inf_{\tau,\delta,\psi_\tau}\ \mathbb{E}_i\{\tau\}}{\inf_{\delta}\ \mathbb{E}_i\{\tau_o\}}=1 \ ,
\end{align}
where $\tau_o$ is the stopping time when the bus selection rules are given in \eqref{eq:op} and \eqref{eq:op2}.
\end{theorem}
\begin{proof}
See Appendix \ref{App:thm2}.
\end{proof}
\noindent Next, by leveraging the results of theorems \ref{thm1} and \ref{thm2} we provide an optimal bus selection rule for the general setting with arbitrary number of anomalous events, $M$, and number of measurements taken at-a-time, $\ell$. 

\subsubsection{Implementation}

Inspired by the results for the binary setting ($M=1$), we devise data-adaptive bus selection rules that can accommodate any arbitrary number of anomalous events $M$. For this purpose, corresponding to each event $R_k$, we assign the following time-varying metric to each bus $i$
\begin{align}
M_i^k(t)\;=\; &\; \frac{1}{2} \sum_{j\in\psi_{t-1}} \log\frac{1}{1-(r^k_{ij})^2}+(r^k_{ij})^2(\Delta\theta_j^2-1) \nonumber\\
+\; & \frac{1}{2|\mathcal{S}_t^i|}\sum_{j\in\mathcal{S}_t^i}\log\frac{1}{1-(r^k_{ij})^2}\ ,
\end{align}
where we have defined
\begin{align}\label{eq:r^k}
\beta^k_i  \;\dff\; \bigg(\sum_{(i,j)\in\E_k} \frac{1}{x^k_{ij}}\bigg)^{-1} \ ,\quad\text{and }\quad r^k_{ij} & \;\dff\; \frac{\beta^k_i}{x^k_{ij}} \ .
\end{align}
Based on these metrics, the optimal data-adaptive bus selection  rule at time $t$ involves selecting the buses that render the largest values in the set
\begin{align}
\big\{M_i^k(t)\; :\ k\in \{0,\dots,M\}\ \ \text{and}\ \ i\in\N\setminus\psi_{t-1} \big\} \ .
\end{align}
Inspired by the observation in Theorem \ref{thm2}, we characterize a simple rule for implementing these selection rules. We first note that for a fixed time $t$ and bus $i$ metric $M_i^k(t)$ takes relatively similar values under different events. In other words, the dynamic range of the set
\begin{align}
\big\{M_i^0(t),\dots,M_i^M(t)\big\}
\end{align}
is very narrow. This is primarily due to the fact that each event $R_k$ only affects a limited number of buses, and consequently, the effects on $M_i^k(t)$ are minor. Motivated by reducing the computational complexity, for each bus $i$, we retain only metric $M_i^0(t)$ as a representative for the set $\{M_i^0(t),\dots,M_i^M(t)\}$. This leads to assigning only one metric to bus $i$ at time $t$ denoted by
\begin{align}\label{eq:busmetric}
M_i(t)\;=\; &\; \frac{1}{2} \sum_{j\in\psi_{t-1}} \log\frac{1}{1-r^2_{ij}}+r^2_{ij}(\Delta\theta_j^2-1) \nonumber\\
+\; & \frac{1}{2|\mathcal{S}_t^i|}\sum_{j\in\mathcal{S}_t^i}\log\frac{1}{1-r_{ij}^2}\ .
\end{align}
This metric consists of three terms, where the first and third terms are functions of the correlation structure through $\{r_{ij}\}$. While $M_i(t)$ as defined in \eqref{eq:busmetric} can be used directly for the bus selection, we offer an alternative two-stage selection rule in order to place more emphasis on the data. In this two-stage approach, in the first stage we focus on the buses that are already observed, and identify the buses whose measurements have the largest level of deviation from the expected values, i.e., the buses with largest $|\theta_i-\bar\theta_i|$. This provides an estimate of the location of the underlying anomaly event, and is equivalent to maximum likelihood decision about the true hypothesis model. In the second stage, among the neighbors of the buses with larger $|\Delta\theta|$, we identify buses with the largest metric $M_i(t)$. Also, at $t=1$, data collection is initialized by selecting $\ell$ buses with the most number of neighbors such that the most informative measurements are collected. The steps of bus selection rule are presented in Algorithm~\ref{table:sel}. Figure~\ref{fig:selection} illustrates the bus selection process for IEEE $14$-bus system under the outage of the line connecting buses $9$ and $14$. By setting $\ell=2$, buses $4$ and $6$, which have the highest degree in this system, are selected at time $t=1$. Since bus $4$ is a neighbor of bus $9$, it experiences larger deviation in its voltage phasor value. Therefore, at time $t=2$ among the neighbors of bus $4$, which are buses $\{2,3,5,7,9\}$, the two buses with the largest metric values are selected, and the process continues in this way.
\begin{figure}[t]
\centering
\includegraphics[width=0.4\textwidth]{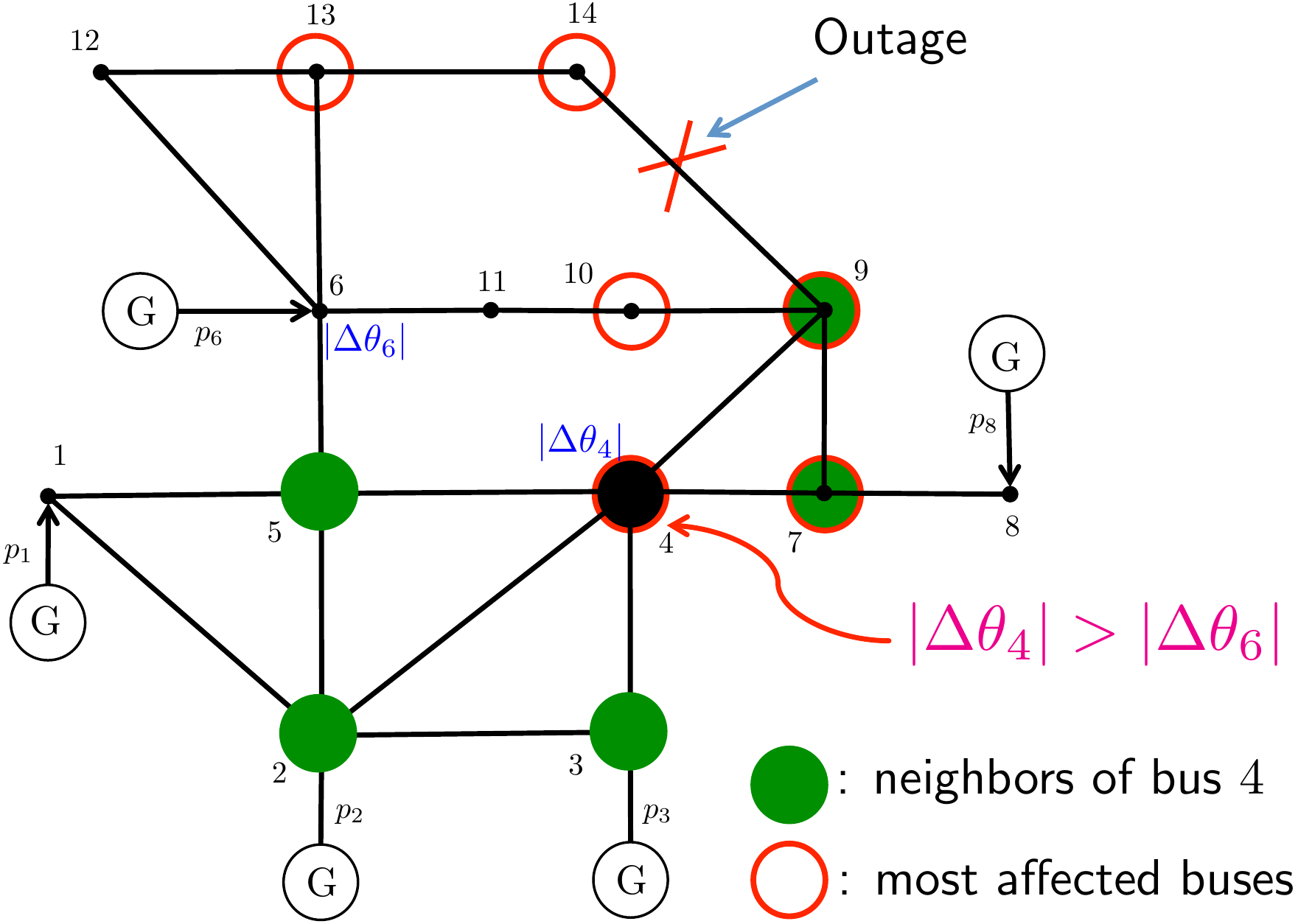}
\renewcommand{\figurename}{Fig.}
\caption{Bus selection process in IEEE $14$-bus system under the outage of line connecting buses $9$ and $14$.}
\label{fig:selection}
\end{figure}

Since the connectivity degree of the  graph underlying the grids is substantially smaller than the size of the grid (e.g., in the IEEE $118$-bus  model the degree is $12$), the complexity of the proposed bus selection rule is substantially lower than that of the exhaustive search.

{\footnotesize
\begin{table}[!h]
\renewcommand{\arraystretch}{1}
\renewcommand{\tablename}{Algorithm}
\footnotesize
\captionsetup{justification=centering}
\caption{Data-adaptive bus selection}
\centering
\begin{tabular}{ll}
\hline\hline
1 & \textbf{Set} $t=0$ and $\cT=\{1,\dots,N\}$\\
2 & \textbf{For} $i=1,\dots,N$ \textbf{repeat}\\
3 & \quad deg$(i)\leftarrow$ Number of buses connected to bus $i$\\
4 & \quad $M(i)\leftarrow\max_{\mathcal{U}\subseteq\N_i}\frac{1}{|\mathcal{U}|}\sum_{j\in\mathcal{U}}\log \frac{1}{1-r^2_{ij}}$\\
5 & \textbf{End for}\\
6 & $\cT\leftarrow$ Sorted $\cT$ based on decreasing deg$(\cdot)$\\
7 & $\psi(t)\leftarrow$ First $\ell$ elements of $\cT$\\
8 & \textbf{While} stopping criterion is not met \textbf{do}\\
9 & \quad Take measurements from buses in $\psi(t)$\\
10 & \quad $S\leftarrow\psi_t$\\
11 & \quad $t \leftarrow t+1$\\
12 & \quad $\psi(t)\leftarrow\{\}$\\
13 & \quad \textbf{While} $|\psi(t)|<\ell$ \textbf{do}\\
14 & \quad\quad $i\leftarrow\argmax_{j\in S} |\theta_j-\bar\theta_j|$\\
15 & \quad\quad $\N_i\leftarrow$ Neighbors of $i$ sorted based on decreasing $M(\cdot)$\\
16 & \quad\quad \textbf{If} $|\N_i|<\ell-|\psi(t)|$ \textbf{then}\\
17 & \quad\quad\quad $\psi(t)\leftarrow\psi(t)\cup\N_i$\\
18 & \quad\quad \textbf{Else}\\
19 & \quad\quad\quad $\psi(t)\leftarrow\psi(t)\cup\{\N_i(1),\dots,\N_i(\ell-|\psi(t)|)\}$\\
20 & \quad\quad \textbf{End if}\\
21 & \quad\quad $S\leftarrow S\setminus i$\\
22 & \quad \textbf{End while}\\
23 & \textbf{End while}\\
24 & Set $\tau=t$\\
\hline\hline
\end{tabular}
\label{table:sel}
\end{table}
}

\subsection{Stopping Time and Decision Rule}

The data-acquisition process is terminated as soon as a decision can be made with the desired reliability, i.e., error probability is controlled below $\beta$. To formalize this, we define $\bn\dff \bp-\bar\bp$
as the perturbations in the power injection incurred by an anomaly, which can be modeled as a zero-mean Gaussian random vector \cite{Zhu:PS12}. We denote the covariance matrix of $\bn$ under $R_k$ by $\bm\Sigma_k$. Hence, based on \eqref{eq:matrix}, under event $R_k$ we have
\begin{equation}\label{eq:prepost}
\bH\bar\bt+\bn=\bH_k\bt\ ,\quad\mbox{for}\ k\in\{0,\dots,M\}\ .
\end{equation}
We denote the incident matrix of the grid by $\bM\in\mathbb{R}^{N\times L}$, which is constructed based on the topology of the grid when there exists no anomaly in the following form. The $i$-th column of $\bM$, denoted by $\blm_i$, corresponds to line $i\in\L$ and all its entries are zero except at two locations that specify the buses connected by line $i$. Specifically, if line $i$ connects buses $m$ and $n$, then the $m$-th and $n$-th entries of $\blm_i$ are $+1$ and $-1$, respectively. Hence, matrix $\bH_k$ can also be generated from the incident matrix of the network and the reactance of transmission lines as follows.
\begin{align}\label{eq:X}
\bH_k & =\sum_{i\in{\L}} X_k[ii]{\blm_{i}}{\blm_{i}}^T= \bM\bX_k\bM^T\ ,
\end{align}
where $\bX_k\in\mathbb{R}^{L\times L}$ is a diagonal matrix defined such that when the $i$-th transmission line connects buses $m$ and $n$ we have $X_k[ii]=\dfrac{1}{x^k_{mn}}$. Hence, \eqref{eq:X} implies that 
\begin{align}\label{eq:hhat}
\bH_k=\bH-\sum_{i\in R_k} \big(X_0[ii]-X_k[ii]\big){\blm_{i}}{\blm_{i}}^T\ ,
\end{align}
which in conjunction with \eqref{eq:prepost} yields
\begin{align}\label{eq:hdt}
\bH\cdot{\Delta\bt} &= \sum_{i\in R_k} \big(X_0[ii]-X_k[ii]\big)\blm_{i}\blm_{i}^T \bt+\bn = \bM\bs_k+\bn\ , \nonumber
\end{align}
where $\bs_k\in\mathbb{R}^L$ is defined as
\begin{equation}
\bs_k[i]\dff\left\{\begin{array}{ll}
\vspace{2mm}
\big(X_0[ii]-X_k[ii]\big)\blm_{i}^T \bt & \mbox{if}\ i\in R_k\\
0 & \mbox{Otherwise}
\end{array}\right.\ .
\end{equation}
The locations of the non-zero elements of vector $\bs_k$ correspond to the indices of the anomalous lines. By assuming that each anomalous event affects a small fraction of the total number of transmission lines, $\bs_k$ becomes a sparse vector. 
Now, by defining $\bB\dff\bH^{-1}$ we obtain
\begin{align}\label{eq:CS}
{\Delta}\bt=\bB\bM\bs_k+\bB\bn\ .
\end{align}
Hence, at the stopping time $\tau$ we have
\begin{align}\label{eq:subset}
{\Delta}\bt_\tau &=\bB_{\tau}\bM\bs_k+\bB_{\tau}\bn\ , 
\end{align}
where $\bB_{\tau}$ is the matrix constructed from $\bB$ by keeping its rows corresponding to set $\psi_\tau$. Since the noise vector $\bB_{\tau}\bn$ is colored, we include a pre-processing whitening stage. For this purpose, we consider the following singular value decomposition (SVD) of matrix $\bB_{\tau}\bm\Sigma_k^{\frac{1}{2}}$, where $\bm\Sigma_k$ is the covariance matrix of $\bn$ under event $R_k$:
\begin{align}
\bB_{\tau}\bm\Sigma_k^{\frac{1}{2}}=\bU_k\bm\Lambda_k \bV_k^T\ .
\end{align}
Then, by defining 
\begin{align}
\by_k &\dff\bm{\Lambda}_k^{-1}\bU_k^T {\Delta}\bt_\tau \ ,\\
\tilde\bn_k &\dff \bm{\Lambda}_k^{-1}\bU_k^T\bB_{\tau}\bn \ ,\\
\label{eq:def}
\text{and }\quad \bA_k&\dff\bV_k^T\bM\ ,
\end{align}
from \eqref{eq:subset}--\eqref{eq:def}, corresponding to event $R_k$ we obtain
\begin{align}\label{eq:CS2}
\by_k=\bA_k\bs_k+\tilde\bn_k\ ,
\end{align}
where $\tilde\bn_k$ is a white noise vector with covariance matrix $\bI$. This leads to an overcomplete representation of sparse vector $\bs_k$ by measurement vector $\by_k$ given in \eqref{eq:CS2}. Therefore, off the shelf tools from compressed sensing can be applied to find the non-zero elements of $\bs_k$ to detect and localize any anomaly event. In this paper, we use orthogonal matching pursuit (OMP) as a fast sparse recovery algorithm, which is summarized in Algorithm~\ref{table:OMP}. In Algorithm~\ref{table:OMP}, the $i$-th column of matrix $\bA_k$ is denoted by vector $\ba_{k,i}$, and the matrix composed of a set of columns of matrix $\bA_k$ indexed in set $\cT$ is denoted by $\bA_{k,\cT}$. The value of threshold $\gamma$ depends on the power of perturbation noise and the performance accuracy constraint $\beta$. It is noteworthy that parameter $\beta$ is set according to the error margin that the network operator can tolerate in localizing the anomalous events. Both $\beta$ and $\gamma$ can be calculated based on some historical data or through a comprehensive simulation of power grid under different events.
{\footnotesize
\begin{table}[h]
\renewcommand{\arraystretch}{1}
\renewcommand{\tablename}{Algorithm}
\footnotesize
\captionsetup{justification=centering}
\caption{OMP Algorithm}
\centering
\begin{tabular}{ll}
\hline\hline
1 & \textbf{Inputs} $\by_k$ and $\bA_k$ for $k\in\{0,\dots,M\}$\\
2 & \textbf{Set} $\br_k=\by_k$, $\bs_k=\bm{0}$ and $\cT_k=\{\}$\\
3 & \textbf{While} $\min_k\|\br_k\|>\gamma$ \textbf{and} $|\cT_1|<\eta_{\rm max}$ \textbf{do}\\
3 & \qquad \textbf{For} $k=1,\dots,M$ \textbf{Repeat}\\
4 & \qquad\qquad $I_k\leftarrow\displaystyle\argmax_{i} \dfrac{|\ba_{k,i}^T\br_k|}{|\ba_{k,i}|}$ \\
4 & \qquad\qquad $\cT_k\leftarrow\cT_k \cup\{I_k\}$\\
5 & \qquad\qquad $\bs_k[\cT_k]=\big(\bA_{k,\cT_k}^T \bA_{k,\cT_k}\big)^{-1}\bA_{k,\cT_k}^T \by_k$\\
6 & \qquad\qquad $\br_k=\by_k-\bA_k\bs_k$\\
7 & \qquad\textbf{End for}\\
7 & \textbf{End while}\\
8 & \textbf{If} $\min_k\|\br_k\|>\gamma$\\
9 & \qquad \textbf{Continue} sampling\\
10 & \textbf{Else}\\
11 & \qquad $i\leftarrow \argmin_k\|\br_k\|$ \\
11 & \qquad \textbf{Stop} sampling and \textbf{Return} $\bs_i$\\
12 & \textbf{End if}\\
\hline\hline
\end{tabular}
\label{table:OMP}
\end{table}
}

This data-adaptive data acquisition and decision-making strategy works based on the offline  and real-time information from the grid. The offline information includes the network topology, the nominal values for the voltage phasor angles in the fault-free situation as well as each anomalous event, which are computed based on the network topology and historical data. The real-time data are the information collected from the buses during the data gathering process.

\section{Case Study: Line Outage Detection}
\label{sec:sim}

When a transmission line is overloaded, the protection devices of the grid automatically remove that line to prevent major damages to the grid and electrical devices. Line outage can be considered a special case of anomaly in the grid, and the devised algorithm can be applied to detect and localize them. For the simulations, we use the software toolbox MATPOWER to generate synthetic data for voltage phasor angles under different outage events \cite{MATPOWER}. In the simulations and by using the IEEE standard systems, we compare the results of data-adaptive data-acquisition approach with the pre-specified bus selection method  in terms of decision accuracy and the number of required measurements. We also evaluate the interplay among delay, number of measurements, and decision accuracy.

\subsection{Gains of Dynamic Bus Selection}

The proposed approach aims to detect and localize the lines under outage with the minimum number of measurements to achieve a target reliability level. The major feature of this approach is data-adaptive selection of buses for acquiring the measurements that are most informative about the state of the grid. In order to assess the gain of such dynamic bus selection, we compare the performance of our proposed approach with the pre-specified bus selection rule in the $118$-bus IEEE standard system. To this end, we fix the number of measurements in the pre-specified bus selection rule and the data-adaptive approach to be the same and compare
their accuracy in localizing the underlying event. In the pre-specified method we select the set of buses with the most number of neighbors, and for the data-adaptive technique we set $\ell=5$. Figure \ref{fig:Fig1} compares the accuracy performance of both methods under the single line outage setting when $\R$ is the set of all single line outage events in which the network is still connected. It is assumed that the perturbation noise vector is uncorrelated with power $1\%$ of the average injected power before any outage, and the number of lines under outage is known. It is observed that for equal number of measurements, the data-adaptive approach uniformly
outperforms the pre-specified method. The reason is that in data-adaptive approach, the correlation structure among the measurements is exploited judiciously to collect measurements from more relevant buses that provide more relevant information about the underlying outage event. Also, it is observed that in the data-adaptive approach the performance  gains diminishes as the number of measurements exceeds $70$, which indicates that by partially observing the grid we can achieve a performance close to the performance of full observation.

\setcounter{table}{0}

\begin{table}[b]
\centering
\captionsetup{justification=centering}
\caption{Average running time comparison.}
\renewcommand{\arraystretch}{1.5}
\begin{tabular}{|c||c|c|c|c|}
\hline
\multicolumn{1}{|c||}{Number of measurements} & $30$ & $50$ & $70$ & $90$ \\ 
\hhline{|=||=|=|=|=|}
$t_{\rm DA}(sec)$ & $0.2774$  & $0.8069$  & $1.4681$  & $2.0817$ \\
\hline
$t_{\rm ES}(sec)$  & $6.4415$  & $14.4312$  & $20.8759$ & $23.398$ \\
\hline
$t_{\rm ES}/t_{\rm DA}$  & $23.2$  & $17.9$  & $14.2$ & $11.2$ \\
\hline
\end{tabular}

 \label{tab:sim1}
\end{table}

Figure~\ref{fig:Fig2} compares the performance for different number of lines in outage under the same settings as in Fig. \ref{fig:Fig1}. We assume that multiple line outages is a result of the overloading of neighboring lines when a single outage occurs. Hence, the lines under outage are in the same locality of the grid. Motivated by Fig. \ref{fig:Fig1}, we set the number of measurements in both methods to $70$ and also include the results for full observation of the network. It is observed that the data-adaptive approach, for single and multiple line outage events, outperforms the pre-specified method by a considerable margin and its performance, as  expected, is close to the full observation of the network.

\begin{figure}[t]
\centering
\includegraphics[width=3.2in,height=2.1in]{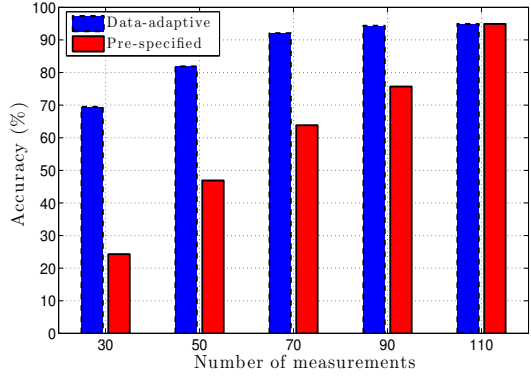}
\renewcommand{\figurename}{Fig.}
\captionsetup{justification=centering}
\caption{Decision accuracy versus number of measurements.}
\label{fig:Fig1}
\end{figure}

\begin{figure}[t]
\centering
\includegraphics[width=3.2in,height=2.1in]{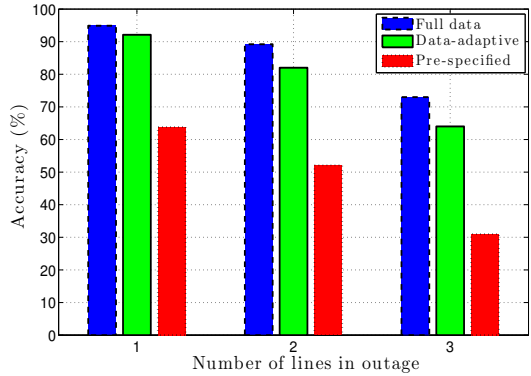}
\renewcommand{\figurename}{Fig.}
\captionsetup{justification=centering}
\caption{Decision accuracy versus number of lines under outage.}
\label{fig:Fig2}
\end{figure}

In order to assess the computational advantage of the proposed approach for bus selection, established in Theorem \ref{thm1}, we compare the simulation time required for implementing the proposed approach with an exhaustive search for finding the most relevant buses. We denote the average of the simulation time over all possible single line outage events for the exhaustive search and the data-adaptive search by $t_{\rm ES}$ and $t_{\rm DA}$, respectively. The results, provided in Table \ref{tab:sim1}, show that for $70$ measurements, which performs close to observing the entire network, data-adaptive collection of measurements is $14.2$ times faster.

\subsection{Trade-off Among Performance Measures}

We consider single line outage setting and evaluate the interplay among different performance measures by changing $\ell$ and the target decision quality $\beta$. In Fig. \ref{fig:dif_ell} the number of required measurements to achieve a certain accuracy level is compared for different values of $\ell$. It is observed that as $\ell$ increases we need more measurements to achieve the same decision accuracy, because larger $\ell$ means taking more measurements at the same time and they cannot incorporate the information of the current time instant. In other words, for $\ell=1$ we take one measurement based on the entire past measurements while in $\ell=5$ for all $5$ new measurements we use the same  information. Also, it is observed that the number of required measurements for improving accuracy from $60\%$ to $70\%$ is less than the one required for improving accuracy from $70\%$ to $80\%$. In order to evaluate the impact of $\ell$ on data collection delay, which is the number of time steps required to collect all the measurements, in Fig.~\ref{fig:Delay_vs_ell} we compare average delay for various $\ell$ and different detection accuracy levels. It is observed that for smaller $\ell$, improving detection accuracy incurs more delay compared to larger $\ell$. Furthermore, for smaller $\ell$ and the same accuracy performance, decreasing $\ell$ leads to more delay.

\begin{figure}[t]
\centering
\includegraphics[width=3.2in,height=2.1in]{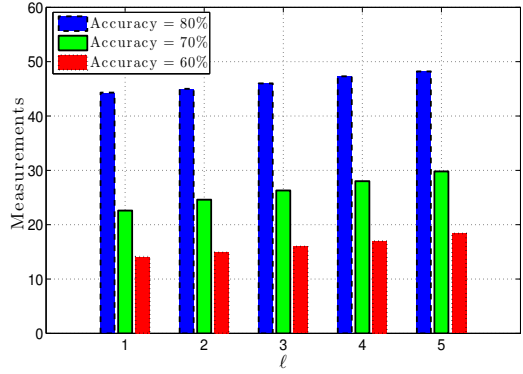}
\renewcommand{\figurename}{Fig.}
\captionsetup{justification=centering}
\caption{Number of measurements versus $\ell$ for different accuracy level.}
\label{fig:dif_ell}
\end{figure}

\begin{figure}[t]
\centering
\includegraphics[width=3.2in,height=2.1in]{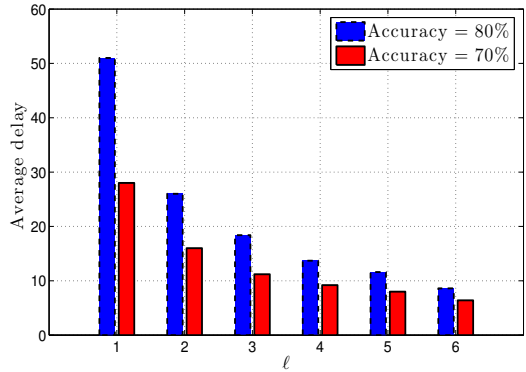}
\renewcommand{\figurename}{Fig.}
\captionsetup{justification=centering}
\caption{Average delay versus $\ell$ for different accuracy level.}
\label{fig:Delay_vs_ell}
\end{figure}

\subsection{Scalability and Complexity}

In order to evaluate the scalability of the proposed detection algorithm, we consider the Polish power system provided by MATPOWER ``case2383wp'' casefile which is a $2383$-bus system. We set $\ell=10$ and consider the noise-free case. The performance of the proposed selection approach for different number of measurements is compared with pre-specified selection rule in Fig.~\ref{fig:Fig6}. It is observed that even for large-scale power systems, the data-adaptive selection rule can achieve considerable performance by selecting a subset of buses in the grid, and outperforms the pre-specified selection approach by a large margin. In fact, as the grid size grows the performance gain improves too. This is primarily due to the fact that larger grids provide more freedom for selecting the buses.

\begin{figure}[t]
\centering
\includegraphics[width=3.2in,height=2.1in]{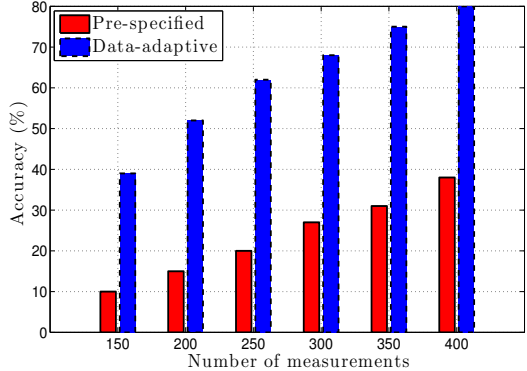}
\renewcommand{\figurename}{Fig.}
\captionsetup{justification=centering}
\caption{Decision accuracy versus number of measurements for noise-free case.}
\label{fig:Fig6}
\end{figure}

\section{Conclusion}

The problem of detecting and localizing anomalies in transmission lines by using the minimum number of measurements has been considered. By adopting a stochastic graphical model for the voltage phasor angles, a data-adaptive strategy for coupled data-acquisition and decision-making processes is designed. Specifically, in this graphical framework the grid connectivities impose a correlation structure among the measurements from different buses. Corresponding to each possible anomalous event, the underlying correlation structure takes a specific form according to the associated topology and parameters of the grid. Hence, depending on the true correlation model, the measurements collected from different buses have different information quality. Data-adaptive monitoring of the network proposed in this paper identifies the most informative buses under each event and minimizes the number of required measurements for a reliable decision about the existing anomaly. A case study for line outage detection and localization confirms the gains of the proposed approach in the IEEE $118$-bus system as well as a $2383$-bus system.

\appendix

\section{Proof of Theorem \ref{thm1}}
\label{App:thm1}

To prove this, we consider a node $v\notin\mathcal{U}_t^i$ at time $t$ and show that one of these cases occurs for the selection of node $v$ in the future:
\begin{enumerate}
\item it will be independent of the sample taken from node $i$ at time $t$, i.e., the data observed from node $i$ has no impact on the information that will be acquired by observing node $v$ in the future; or,
\item despite dependence of its information on the sample taken from node $i$, the amount of this information will be less than the expected information of observing the best subset of $\mathcal{U}_t^i$, i.e.,
\begin{align}\label{eq:proof1}
\max_{\mathcal{S}_t^i\subseteq\mathcal{U}_t^i} \frac{M_i^{\ell}(t,\mathcal{S}_t^i\cup\{v\})}{|\mathcal{S}_t^i|+1}\leq\max_{\mathcal{S}_t^i\subseteq\mathcal{U}_t^i} \frac{M_i^{\ell}(t,\mathcal{S}_t^i)}{|\mathcal{S}_t^i|}\ .
\end{align}
\end{enumerate}
We consider all scenarios for node $v\notin\mathcal{U}_t^i$ and show how each scenario falls into one of these two categories. Since the graph is connected, there exists a node $j\in\N_i$ which belongs to the path between nodes $i$ and $v$. If node $j$ belongs to the set $\mathcal{S}_t^i$ that maximizes the right hand side of \eqref{eq:proof1}, due to the global Markov property, $i$ and $v$ will be conditionally independent, which makes case $1$ true.

Now, we only need to show that whenever node $j$ is outside the set $\mathcal{S}_t^i$ that maximizes the right hand side of \eqref{eq:proof1}, inclusion of node $v$ will reduce the average information. To this end, we note that the marginal distribution of each random variable $\theta_i$ under both hypothesis is the same. We prove that the information of observing node $j$ is greater than that of observing node $v$ which means that if $j$ does not belong to the set that maximizes the normalized Kullback-Leibler (KL) divergence, node $v$ should not be in that set, too. For this purpose, by denoting the probability density function under $\H_i$ by $f_i$ and the KL divergence between $f_0$ and $f_1$ by $D_{\textsc{\tiny KL}}(f_0\|f_1)$, we compute $D_{\textsc{\tiny KL}}(f_0(\theta_i,\theta_j,\theta_v)\|f_0(\theta_i)f_0(\theta_j,\theta_v))$ by following two different strategies and compare the results. 
\begin{align}\label{eq:H1}
D &_{\textsc{\tiny KL}} (f_0(\theta_i,\theta_j,\theta_v)\|f_0(\theta_i)f_0(\theta_j,\theta_v)) \nonumber \\
&= D_{\textsc{\tiny KL}} (f_0(\theta_i,\theta_j)f_0(\theta_v|\theta_j)\|f_0(\theta_i)f_0(\theta_j)f_0(\theta_v|\theta_j)) \nonumber \\
&= D_{\textsc{\tiny KL}} (f_0(\theta_i,\theta_j)\|f_0(\theta_i)f_0(\theta_j)) \nonumber \\
&= D_{\textsc{\tiny KL}} (f_0(\theta_i,\theta_j)\|f_1(\theta_i)f_1(\theta_j)) \ .
\end{align}
On the other hand
\begin{align}\label{eq:H1_1}
D_{\textsc{\tiny KL}} & (f_0(\theta_i,\theta_j,\theta_v)\|f_0(\theta_i)f_0(\theta_j,\theta_v)) \nonumber \\
=\,& D_{\textsc{\tiny KL}} (f_0(\theta_i,\theta_v)f_0(\theta_j|\theta_i,\theta_v)\|f_0(\theta_i)f_0(\theta_v)f_0(\theta_j|\theta_v)) \nonumber \\
=\,& D_{\textsc{\tiny KL}} (f_0(\theta_i,\theta_v)\|f_0(\theta_i)f_0(\theta_v)) \nonumber \\
& + D_{\textsc{\tiny KL}} (f_0(\theta_j|\theta_i,\theta_v)\|f_0(\theta_j|\theta_v)) \nonumber \\
\geq\,& D_{\textsc{\tiny KL}} (f_0(\theta_i,\theta_v)\|f_0(\theta_i)f_0(\theta_v)) \nonumber \\
=\,& D_{\textsc{\tiny KL}} (f_0(\theta_i,\theta_v)\|f_1(\theta_i)f_1(\theta_v)) \ ,
\end{align}
where the inequality holds due to the non-negativity of KL divergence. Since the left hand side of \eqref{eq:H1} and \eqref{eq:H1_1} are the same, we have
\begin{align}\label{eq:procc}
D_{\textsc{\tiny KL}} (f_0(\theta_i,\theta_j)\| & f_1(\theta_i)f_1(\theta_j)) \geq \nonumber \\ & D_{\textsc{\tiny KL}} (f_0(\theta_i,\theta_v)\|f_1(\theta_i)f_1(\theta_v)) \ .
\end{align}
By following the same line of argument for computation of $D_{\textsc{\tiny KL}}(f_0(\theta_i)f_0(\theta_j,\theta_v)\|f_0(\theta_i,\theta_j,\theta_v))$ we obtain
\begin{align}\label{eq:procc2}
D_{\textsc{\tiny KL}} (f_1(\theta_i) & f_1(\theta_j)\| f_0(\theta_i,\theta_j)) \geq \nonumber \\ & D_{\textsc{\tiny KL}} (f_1(\theta_i)f_1(\theta_v)\|f_0(\theta_i,\theta_v)) \ .
\end{align}
From \eqref{eq:procc} and \eqref{eq:procc2} we can conclude that the divergence of two distributions is maximal between neighbor nodes, which concludes the proof.

\section{Proof of Theorem \ref{thm2}}
\label{App:thm2}

In order to prove this theorem we first assume that the size of the network grows to infinity and for any sequence of buses $\psi_t$ and any sequence of measurements from those buses $\bt_t$,
\begin{align}
\label{eq:cc1}&\frac{1}{N}\log\frac{f_0(\bt_N;\psi_N)}{\prod_{i\in\psi_N} f_1(\theta_i)}\rightarrow I_0\ ,\quad\text{under }R_0\ ,\\
\label{eq:cc2}\text{and }\quad &\frac{1}{N}\log\frac{\prod_{i\in\psi_N} f_1(\theta_i)}{f_0(\bt_N;\psi_N)}\rightarrow I_1\ ,\quad\text{under }R_1\ ,
\end{align} 
converge completely as ${N\rightarrow\infty}$. We define 
\begin{align*}
\alpha_0&=\mathbb{P}(\delta=0|\T=1)\ ,\\
\text{and}\quad\alpha_1&=\mathbb{P}(\delta=1|\T=0)\ , 
\end{align*}
and show that for the stopping time of the optimal strategy we have 
\begin{align}\label{eq:E0}
\lim_{\mathsf{P_e}\rightarrow0}\inf_{\tau,\psi_{\tau}} \ \bbe_0\{\tau\}
&\geq\frac{|\log\alpha_0|}{I_0}\ ,\\
\label{eq:E1}
\mbox{and}\ ,\quad\lim_{\mathsf{P_e}\rightarrow0}\inf_{\tau,\psi_{\tau}} \ \bbe_1\{\tau\}
&\geq\frac{|\log\alpha_1|}{I_1}\ ,
\end{align}
and then we prove that the selection rule designed in~\eqref{eq:op} and~\eqref{eq:op2} achieves these lower bounds. If we show that for all $0<\rho<1$
\begin{align}\label{eq:p0}
&\lim_{\mathsf{P_e}\rightarrow0}\inf_{\tau,\psi_{\tau}} \ \mathbb{P}_0\left(\tau>\rho\frac{|\log\alpha_0|}{I_0}\right)=1\ ,\\
\label{eq:p1}
\mbox{and}\ ,\quad &\lim_{\mathsf{P_e}\rightarrow0}\inf_{\tau,\psi_{\tau}} \ \mathbb{P}_1\left(\tau
>\rho\frac{|\log\alpha_1|}{I_1}\right)=1\ ,
\end{align}
then by applying the generalized Chebyshev inequality, we obtain
\begin{align}
\mathbb{E}_0\left\{\frac{\tau}{\frac{|\log\alpha_0|}{I_0}}\right\} & \geq\rho\cdot\mathbb{P}_0\left(\frac{\tau}{\frac{|\log\alpha_1|}{I_0}}>\rho\right)\overset{\eqref{eq:p0}}{\geq}1\ ,\quad \forall \rho>0 
\end{align}
which concludes \eqref{eq:E0}. By following the same line of arguments \eqref{eq:E1} will be proved.

Now, we prove \eqref{eq:p1} and the procedure for \eqref{eq:p0} will follow the same line of thought. Let us define the event
\begin{align}
\A(i,J)\dff\{\delta=i\, ,\, \tau\leq J\} \ ,
\end{align}
and the log-likelihood ratio
\begin{equation}\label{LLR}
\Lambda_t \;\dff\; \log\frac{\prod_{s=1}^t f_1(\theta_s;\psi(s))}{f_0(\bt_t;\psi_t)}\ .
\end{equation}
Then, by Wald's identity and for any $0<J<n$ and $B>0$ we have
\begin{align}
\alpha_1 & =\mathbb{P}_0(\delta=1) \nonumber\\
& = \mathbb{E}_0\{\mathds{1}_{(\delta=1)}\}\nonumber\\
& = \mathbb{E}_1\{\mathds{1}_{(\delta=1)}\exp(-\Lambda_{\tau})\}\nonumber\\
& \geq \mathbb{E}_1\{\mathds{1}_{(\A(1,J),\Lambda_{\tau}<B)}\exp(-\Lambda_{\tau})\}\nonumber\\
& \geq e^{-B}\mathbb{P}_1(\A(1,J),\Lambda_{\tau}<B)\nonumber\\
& \geq e^{-B}\mathbb{P}_1\Big(\A(1,J),\sup_{t<J}\Lambda_t<B\Big)\nonumber\\
& \overset{(a)}{\geq}e^{-B}\Big(\mathbb{P}_1(\A(1,J))-\mathbb{P}_1\big(\sup_{t<J}\Lambda_t\geq B\big)\Big)\nonumber\\
&\overset{(b)}{\geq}e^{-B}\Big(\mathbb{P}_1(\delta=1)-\mathbb{P}_1(\tau>J)-\mathbb{P}_1\big(\sup_{t<J}\Lambda_t\geq B\big)\Big)\nonumber \ ,
\end{align}
where (a) and (b) hold due to the properties of set difference operation. Now we have
\begin{align}\label{eq:p1L}
\mathbb{P}_1(\tau>J) & \geq \mathbb{P}_1(\delta=1)-e^B \mathbb{P}_0(\delta=1)-\mathbb{P}_1\big(\sup_{t<J}\Lambda_t\geq B\big) \nonumber\\
& = 1-\alpha_0-e^B\alpha_1-\mathbb{P}_1\big(\sup_{t<J}\Lambda_t\geq B\big)\ .
\end{align}
Since \eqref{eq:p1L} holds for any $B>0$, we set $B=cJI_1$ for some $c>1$. Then for any $1<K<J$ we obtain 
\begin{align}
\mathbb{P}_1&\Big(\sup_{t<J}\Lambda_t\geq B\Big)\nonumber\\
&=\mathbb{P}_1\Big(\sup_{t<J}\Lambda_t\geq cJI_1\Big)\nonumber\\
&\leq\mathbb{P}_1\Big(\sup_{t<K}\Lambda_t+\sup_{K<t<J}\Lambda_t\geq cJI_1\Big)\nonumber\\
&\leq\mathbb{P}_1\Big(\sup_{t<K}\Lambda_t+\sup_{K<t<J}\big(\frac{J}{t}\Lambda_t\big)-JI_1\geq (c-1)JI_1\Big)\nonumber\\
&\leq\mathbb{P}_1\Big(\frac{1}{J}\sup_{t<K}\Lambda_t+\sup_{K<t<J}\big(\frac{\Lambda_t}{t}-I_1\big)\geq (c-1)I_1\Big)\nonumber\\
&\leq\mathbb{P}_1\Big(\frac{1}{J}\sup_{t<K}\Lambda_t+\sup_{t>K}\big|\frac{\Lambda_t}{t}-I_1\big|\geq (c-1)I_1\Big)\ .
\end{align}
According to \eqref{eq:cc2}, for any $\epsilon>0$ there exists a $\hat{K}(\epsilon)$ such that
\begin{align}
\mathbb{P}_1\Big(\big|\frac{\Lambda_{t}}{t}-I_1\big|\leq\epsilon\Big)=1 \ ,\quad \forall t>\hat{K}(\epsilon) \ .
\end{align}
Hence, we have
\begin{align}\label{ineq}
\mathbb{P}_1&\Big(\sup_{t<J}\Lambda_t\geq cJI_1\Big)\leq\mathbb{P}_1\Big(\frac{1}{J}\sup_{t<\hat{K}(\epsilon)}\Lambda_t\geq (c-1)I_1-\epsilon\Big)\ .
\end{align}
Since, this is true for any $J<n$ and $c>1$, we assume the case that $n,J\rightarrow\infty$ and $c>1+\frac{\epsilon}{I_1}$. In this setting the right hand side of \eqref{ineq} approaches zero which indicates that for every $c>1$
\begin{align}\label{zero}
\lim_{L\rightarrow\infty}\mathbb{P}_1&\Big(\sup_{t<J}\Lambda_t\geq cJI_1\Big)=0\ .
\end{align}
Next, for any $0<\rho<\frac{1}{c}$ and  by defining
\begin{align}
J_{\alpha_1}\dff\rho\frac{|\log\alpha_1|}{I_1}\ ,
\end{align}
and setting $J=J_{\alpha_1}$ we obtain
\begin{align}\label{ineq2}
\mathbb{P}_1\bigg(\tau>\rho&\frac{|\log\alpha_1|}{I_1}\bigg) \geq \nonumber \\
& 1-\alpha_0-\alpha_1^{1-\rho c}-\mathbb{P}_1\big(\sup_{t<J_{\alpha_1}}\Lambda_t\geq cI_1J_{\alpha_1}\big)\ .
\end{align}
Now, by combining \eqref{zero} and \eqref{ineq2}, and for the setting in which $\alpha_1$ and $\alpha_0$ approach zero we obtain
\begin{align}\label{ww}
\mathbb{P}_1\left(\tau>\rho\frac{|\log\alpha_1|}{I_1}\right) = 1
\end{align}
Since \eqref{ww} holds regardless of the sampling procedure and stopping rule and only depends on the error performance of the strategy, it is valid for any strategy with the same decision quality, i.e.,
\begin{align}\label{ww2}
\lim_{\alpha_1,\alpha_0\rightarrow0}\inf_{\tau,\psi^{\tau}}\mathbb{P}_1\left(\tau>\rho\frac{|\log\alpha_1|}{I_1}\right) = 1
\end{align}
By following the same line of arguments for $\alpha_0$ and the average delay under $\H_0$,  \eqref{eq:p0} can be proved.

Now, we only require to show that the sequential strategy of this paper achieve the lower bounds on delay given in \eqref{eq:E0} and \eqref{eq:E1}. To this end, first we leverage the properties of complete convergence in \eqref{eq:cc1} and \eqref{eq:cc2}. Specifically, by defining 
\begin{align}\label{T0}
T_0(h)&\dff\sup\ \Big\{t\,:\,\Big|\frac{-\Lambda_t}{tI_0}-1\Big|>h\Big\} \ ,\\
\label{T1}
\mbox{and}\ ,\quad T_1(h)&\dff\sup\ \Big\{t\,:\,\Big|\frac{\Lambda_t}{tI_1}-1\Big|>h\Big\} \ ,
\end{align}
according to the complete convergence of \eqref{eq:cc1} and \eqref{eq:cc2} we have 
\begin{align}
\mathbb{E}_0\{T_0(h)\}<\infty\ , \quad \forall h>0\ ,\\
\mbox{and}\ ,\quad \mathbb{E}_1\{T_1(h)\}<\infty\ , \quad \forall h>0\ .
\end{align}
According to the definition of stopping time 
\begin{align}\label{ineq11}
\Lambda_{\tau-1}<\gamma_U \ .
\end{align}
Also, from \eqref{T1} and when $\tau>T_1(h)+1$ we have
\begin{align}\label{ineq22}
\Lambda_{\tau-1} > (\tau-1)(1-h)I_1 \ .
\end{align}
By combining inequalities in \eqref{ineq11} and \eqref{ineq22} we obtain
\begin{align}
\tau & < 1+\mathds{1}_{(\tau>T_1(h)+1)}\frac{\gamma_U}{I_1(1-h)}\nonumber\\
& \leq 1+\mathds{1}_{(\tau>T_1(h)+1)}\frac{\gamma_U}{I_1(1-h)}+\mathds{1}_{(\tau\leq T_1(h)+1)}T_1(h)\nonumber\\
& \leq 1+\frac{\gamma_U}{I_1(1-h)}+T_1(h)\ .
\end{align}
Therefore, by applying \eqref{T1} we can conclude that
\begin{align}
\mathbb{E}_1\{\tau\}\leq \frac{\gamma_U}{I_1}(1+o(1)) \ ,
\end{align}
and by replacing $\gamma_U$ from 
\begin{align}
\gamma_U = -\log\alpha_1 
\end{align}
we have
\begin{align}
\mathbb{E}_1\{\tau\}\leq \frac{|\log\alpha_1|}{I_1}(1+o(1)) \ .
\end{align}
By following the same line of argument for $\Lambda_{\tau-1}>\gamma_L$ we can derive
\begin{align}
\mathbb{E}_0\{\tau\}\leq \frac{|\log\alpha_0|}{I_0}(1+o(1)) \ ,
\end{align}
which concludes the proof.

\bibliographystyle{IEEEtran}
\bibliography{QD2}

\end{document}